\documentclass[runningheads]{llncs}

\usepackage[T1]{fontenc}
\usepackage[english]{babel}

\usepackage{orcidlink}

\usepackage{mathtools} 
\usepackage{amssymb}

\usepackage[dvipsnames]{xcolor}
\usepackage{graphicx}
\usepackage{tabu}
\usepackage{wrapfig}

\usepackage{tikz}
\usetikzlibrary{automata,positioning}
\usepackage{listings, lstautogobble}

\usepackage{algorithm} 
\usepackage[endLComment={}]{algpseudocodex} 

\usepackage{subcaption}

\usepackage[shortlabels,inline]{enumitem}
\usepackage{dsfont}
\usepackage{xifthen}
\usepackage{xparse}
\usepackage{hyperref}
\usepackage[Export]{adjustbox}
\usepackage{xspace}

\usepackage{tabu}
\usepackage{booktabs}

\usepackage{nicematrix}
\usepackage{xfrac}

\usepackage{environ}

\usepackage[acronym]{glossaries} 
\glsdisablehyper 
\newacronym{mjls}{MJLS}{Markov jump linear system}
\newacronym{pctl}{PCTL}{probabilistic computation tree logic}
\newacronym{pctlm}{PCTL\(^{-}\)}{probabilistic computation tree logic}
\newacronym{ms}{MS}{mean stable}
\newacronym{mss}{MSS}{mean-square stable}

\newacronym[longplural={general Markov processes}]{gmp}{GMP}{general Markov process}
\newacronym{lds}{LDS}{linear (discrete-time) dynamical system}
\newacronym{ltl}{LTL}{linear temporal logic}
\newacronym{mso}{MSO}{monadic second-order logic}
\newacronym{mrm}{MRM}{Markov reward model}
\newacronym{dtmc}{DTMC}{discrete-time Markov chain}

\tikzset{
  every initial by arrow/.style={thick, ->, >=stealth},
  every picture/.style={initial text=},
  >=stealth, 
  every node/.style={line width=0.75pt, align=center}, 
  every path/.style={line width=0.75pt}
}

\newcommand{\NN}{\mathbb{N}}
\newcommand{\RR}{\mathbb{R}}

\newcommand{\mi}[1]{\mathit{#1}}

\providecommand{\where{}} 
\newcommand\SetSymbol[1][]{%
  \nonscript\;#1\vert
  \allowbreak
  \nonscript\;\mathopen{}
}
\DeclarePairedDelimiterX\curly[1]{\lbrace}{\rbrace}{%
  \renewcommand{\where}{\SetSymbol[\delimsize]}#1
}
\newcommand{\set}{\curly*}
\DeclarePairedDelimiterX\abs[1]{\lvert}{\rvert}{%
  \ifblank{#1}{\:\cdot\:}{#1}
}
\DeclarePairedDelimiterX\norm[1]{\lVert}{\rVert}{%
  \ifblank{#1}{\:\cdot\:}{#1}
}
\DeclarePairedDelimiterX\xsemantics[1]{\llbracket}{\rrbracket}{%
  \ifblank{#1}{\:\cdot\:}{#1}
}

\newcommand{\ie}{i.e.\@\xspace}
\newcommand{\eg}{e.g.\@\xspace}

\newcommand{\cf}{cf.\@\xspace}

\newcommand{\repeattheorem}[1]{%
  \begingroup
  \renewcommand{\thetheorem}{\ref{#1}}%
  \expandafter\expandafter\expandafter\theorem
  \csname reptheorem@#1\endcsname
  \endtheorem
  \endgroup
}

\NewEnviron{reptheorem}[1]{%
  \global\expandafter\xdef\csname reptheorem@#1\endcsname{%
    \unexpanded\expandafter{\BODY}%
  }%
  \expandafter\theorem\BODY\unskip\label{#1}\endtheorem
}

\newcommand{\repeatlemma}[1]{%
  \begingroup
  \renewcommand{\thelemma}{\ref{#1}}%
  \expandafter\expandafter\expandafter\lemma
  \csname replemma@#1\endcsname
  \endlemma
  \endgroup
}

\NewEnviron{replemma}[1]{%
  \global\expandafter\xdef\csname replemma@#1\endcsname{%
    \unexpanded\expandafter{\BODY}%
  }%
  \expandafter\lemma\BODY\unskip\label{#1}\endlemma
}

\AddToHook{cmd/appendix/before}{} 

\newcommand{\bnfmid}{\quad|\quad}

\newcommand{\sat}[1]{\text{Sat}\left(\ifthenelse{\isempty{#1}}{\cdot}{#1}\right)}

\newcommand{\ttrue}{\mi{true}}

\NewDocumentCommand{\cylset}{O{\left}O{\right}m}{\mi{Cyl}#1(\ifthenelse{\isempty{#3}}{\cdot}{#3}#2)}
\NewDocumentCommand{\prob}{O{\left}O{\right}m}{\mi{Pr}#1(\ifthenelse{\isempty{#3}}{\cdot}{#3}#2)}
\newcommand{\probmatrix}{P}
\newcommand{\transprob}[2]{\probmatrix\left(\ifthenelse{\isempty{#1}}{\cdot}{#1}, \ifthenelse{\isempty{#2}}{\cdot}{#2}\right)}
\newcommand{\statematrix}[1]{A_{#1}}

\newcommand{\paths}[2][]{\ifthenelse{\isempty{#1}}{\mi{P\!aths}\left(#2\right)}{\mi{P\!aths}_{#1}\left(#2\right)}}

\newcommand{\expectation}[3]{\mathbb{E}_{#1}\left[#2\ifthenelse{\isempty{#3}}{}{\;\middle\vert\; #3}\right]}
\newcommand{\variance}[3]{\mathbb{V}_{#1}\left[#2\ifthenelse{\isempty{#3}}{}{\;\middle\vert\; #3}\right]}

\newcommand{\transpose}{^{\scriptscriptstyle \mathsf{T}}}
\newcommand{\kron}{\mathbin{\otimes}}
\newcommand{\blkdiag}[1]{\text{blkdiag}(#1)}
\newcommand{\indicator}[2]{\mathbf{1}_{#1}\left(#2\right)}

\newcommand{\spectral}[1]{\rho\left(#1\right)}

\newcommand{\summn}[3]{\textsc{sum}^{#1}_{#2}\ifthenelse{\isempty{#3}}{}{\left(#3\right)}}
\newcommand{\mvec}[1]{\textsc{vec}\ifthenelse{\isempty{#1}}{}{\left(#1\right)}}
\newcommand{\unmvec}[2]{\textsc{vec}^{-1}_{#1}\ifthenelse{\isempty{#2}}{}{\left(#2\right)}}
\newcommand{\mvech}[1]{\textsc{vech}\ifthenelse{\isempty{#1}}{}{\left(#1\right)}}
\newcommand{\unmvech}[1]{\textsc{vech}^{-1}\ifthenelse{\isempty{#1}}{}{\left(#1\right)}}

\NewDocumentEnvironment{bcmatrix}{mo}{\colorlet{savedleftcolor}{.}\color{#1}\left[\color{savedleftcolor}\begin{matrix}}{\end{matrix}\color{\IfNoValueTF{#2}{#1}{#2}}\right]}
\NewDocumentEnvironment{bscmatrix}{mo}{\colorlet{savedleftcolor}{.}\color{#1}\left[\color{savedleftcolor}\begin{smallmatrix}}{\end{smallmatrix}\color{\IfNoValueTF{#2}{#1}{#2}}\right]}

\newcommand{\identmatrix}[1]{I_{#1}}
\newcommand{\zeromatrix}[1]{0_{#1}}
\newcommand{\xtheta}{(x, \theta)}
\newcommand{\statespace}[1][\mathfrak{J}]{S_{#1}}
\newcommand{\boperator}{\mathcal{B}}
\newcommand{\toperator}{\mathcal{T}}
\newcommand{\basevect}[1]{\vec{e}_{#1}}
\newcommand{\ap}{\mi{AP}}

\newcommand{\vDashGMP}{\vDash_{\mathsf{GMP}}}
\newcommand{\Q}[2]{Q\left(#1, #2\right)}

\NewDocumentCommand{\PCTLprob}{O{\left}O{\right}mm}{\mathcal{P}_{#3}#1(#4#2)}

\newcommand{\PCTLuntil}[2]{#1 \mathbin{\mathcal{U}} #2}
\newcommand{\PCTLbuntil}[3]{#1 \mathbin{\mathcal{U}^{\leq #3}} #2}
\newcommand{\PCTLnext}[1]{\mathcal{X} \, #1}
\newcommand{\PCTLbexp}[2]{\mathcal{E}^{#1}_{#2}}
\newcommand{\PCTLexp}[1]{\mathcal{E}_{#1}}
\newcommand{\PCTLbvar}[2]{\mathcal{V}^{#1}_{#2}}
\newcommand{\PCTLvar}[1]{\mathcal{V}_{#1}}

\begin{document}

\renewcommand{\sectionautorefname}{Section} 
\renewcommand{\subsectionautorefname}{Section} 
\renewcommand{\subsubsectionautorefname}{Section} 
\renewcommand{\figureautorefname}{Fig.\@}

\title{Stability Checking of Markov Jump Linear Systems via Probabilistic Temporal Logic (Extended Version)}
\titlerunning{Stability Checking of \acrshort{mjls} via Probabilistic Temporal Logic}
%
\author{Lena Becker\orcidlink{0009-0007-8361-3505} \and
Holger Hermanns\orcidlink{0000-0002-2766-9615}}
\authorrunning{L.\@ Becker \and H.\@ Hermanns}
%
\institute{Saarland University, Saarland Informatics Campus, 66123 Saarbrücken, Germany}
\maketitle              
\begin{abstract}
  \glspl{mjls} model dynamical phenomena subject to random switching among multiple linear modes, driven by an underlying Markov chain.
  Classical notions such as mean and mean-square stability characterize the long-term asymptotic behaviour of the first and second moments of an \gls{mjls}, but they can be overly conservative or even misleading when only a specific subset of initial conditions is of interest.
  We tackle this challenge through the lens of model checking, where reasoning about specific sets of initial conditions is intrinsic to the approach. 
  We begin by formalizing \gls{pctl} on \glspl{mjls}, enabling the specification of state-based temporal properties for these systems.
  Building on this foundation, we extend the logic to capture moment-based stability properties relative to a prescribed set of initial states.
  While we ultimately do not obtain a decision procedure for the entire base logic, the logical extensions can be handled -- albeit with some technical subtleties -- by exploiting linear-algebraic techniques.
\end{abstract}

\glsresetall 
\glsunset{pctlm}

\section{Introduction}\label{sec:introduction}

Dynamical systems are a fundamental tool for modelling the evolution of real-world processes over time, typically with the objective of designing controllers that ensure a desired system behaviour.
In many applications, however, the underlying dynamics are subject to uncertainty and may exhibit abrupt changes, \eg in case of system failures or volatile environmental circumstances, inducing switches between different modes of dynamics, each covering the behaviour in a specific situation.
We here focus on models where these switches are probabilistic in nature, and thus governed by a Markov chain.
In such \glspl{mjls}~\cite{costa2005discrete,costa2012continuous}, modes are mapped to distinct system dynamics, adding discrete behaviour to the continuous-state system.
These system dynamics are typically assumed to be linear, as many systems can be locally approximated by linear dynamics near operating points.
\glspl{mjls} have been used to model a wide range of application domains including flight control~\cite{gray2000flight}, power systems analysis~\cite{loparo1990power}, and robotic systems such as lower-limb exoskeletons~\cite{nogueira2014exoskeleton}.

In settings where the system cannot be controlled, or where the controller is already fixed, one must rely on analysis techniques to assess its behaviour.
A commonly analysed property of dynamical systems in this context is their stability, \ie whether the system's state converges in the long run, independently of the initial conditions.
The stability analysis of \glspl{mjls} not only needs to take into account the dynamical behaviour of all system modes, but also the long-run behaviour of the underlying Markov chain.
It may  yield counterintuitive results:
even if each individual dynamics considered in isolation is stable, the \gls{mjls} as a whole may exhibit unstable behaviour, and vice versa.

In the literature, different notions of stability for \glspl{mjls} have been explored, prominently among them \emph{mean stability}~\cite{lian2015mean} -- analysing the evolution of the first moment, and \emph{mean-square stability}~\cite{costa1993stability,costa2005discrete,feng1992stability} -- accounting for cancellation effects due to opposing outliers by considering the second moment as the stability criterion of choice.
Both stability notions are typically characterized via eigenvalue analysis of specific operators that capture the combined continuous and discrete probabilistic dynamics of the \gls{mjls}.
In doing so, it is standard to aim for global guarantees covering all initial states.
This perspective may however be overly conservative in practice.
In particular, instability might arise only for a small fraction of initial states that may not even be of practical relevance, \eg if corresponding to physically unreachable configurations, such as physically impossible operating regions of a drone.

Despite its intuitive importance, the \gls{mjls} stability problem for a given set of initial states has not been tackled in the literature, at least to the best of our knowledge.
It is a genuine model-checking problem, so we attack it by embedding appropriate logical operators into a typical temporal-logic framework.
Probabilistic model checking, based on logics such as \gls{pctl}, enables the specification of temporal properties with respect to specific sets of initial states~\cite{hansson1994logic}.
\gls{pctl} for \glspl{mjls} arises naturally by conservatively extending work on \gls{pctl} for \glspl{dtmc}~\cite{baier2008principles}, as well as by specializing work for \glspl{gmp}~\cite{ramponi2010connections}.
When turning to the model-checking problem for this logic, however, it turns out that the latter work must necessarily be incomplete, because of a Skolem-hardness result~\cite{chonev2016orbit} that applies to a specific fragment of \gls{pctl} model checking on \glspl{mjls} (and thus equally on \glspl{gmp} from~\cite{ramponi2010connections}).
We proceed by extending the logic with operators to capture moment-based stability properties and discuss the intricacies when checking properties that include such operators.
Altogether, we arrive at a logic that enables the specification of refined stability-like properties as well as their combination with classical reachability or branching-time objectives.
A computational example enriches the exposition.

The remainder of this paper is structured as follows:
\autoref{sec:background} develops relevant background on \glspl{mjls}.
In \autoref{sec:pctl}, we define a semantics for classical \gls{pctl} for \glspl{mjls} and give an algorithm to verify such \gls{pctl} properties.
Furthermore, we discuss limitations in model checking reachability properties due to the uncountable state space of \glspl{mjls}.
Finally, in \autoref{sec:expectation}, we extend \gls{pctl} with the aforementioned new operators that capture the first and second moments of \glspl{mjls} and describe an algorithm to verify these properties.

\subsubsection*{Related Work.}
\emph{Discrete-time stochastic systems} that are given by the equation \(x(k+1)=\mathcal{F}_{a}\left(x(k), w(k)\right)\), where \(a\) ranges over the available modes, generalize \glspl{mjls} in multiple aspects:
the functions~\(\mathcal{F}_{a}\) may be non-linear, the switching is not necessarily governed by a Markov chain, and a stochastic noise~\(w(k)\) is incorporated.
Contrasting our approach, existing work on the verification of such systems is based on abstractions into finite-state models~\cite{lahijanian2015verification}.
A \gls{pctl} semantics is then defined on these abstractions and model checking is carried out in a fashion similar to standard algorithms for Markov chains or Markov decision processes.

\emph{\Glspl{mrm}} extend \glspl{dtmc} by associating rewards with the states and accumulating them additively~\cite{andova2004rewards}.
The use of multiplicative rewards instead~\cite{baier2025multiplicative} yields a modelling formalism that closely resembles \glspl{mjls} with a single continuous dimension.
This prevents modelling dependencies across multiple continuous dimensions.
Further, since only positive rewards are supported, behaviours that rely on sign changes or cancellations are excluded.
The authors consider a notion of expectation that is similar in spirit to our treatment of continuous state variables, but oscillatory behaviour is handled by distinguishing between limit superior and limit inferior.
In contrast, we consider a notion based on limits of Cesàro-averaged expected values.

\emph{\Glspl{lds}} given by \({x(k+1)=Ax(k)}\) have been subject to a long line of research investigating the decidability of reachability problems~\cite{karimov2022decidable}.\footnote{The term ``reachability'' here refers to the set of states reachable without considering the option to choose controller inputs in a stepwise way, \ie the orbit problem.}
While deciding reachability of a specific point is possible in polynomial time~\cite{kannan1986orbit}, reachability of more complex sets, such as a four-dimensional affine subspace in a five-dimensional system, can be encoded as longstanding open problems like the Skolem problem~\cite{chonev2016orbit}.
Furthermore, model checking \gls{ltl} is known to be decidable for \glspl{lds} up to dimension three~\cite{karimov2020ltl}.
Since \glspl{lds} can be viewed as a deterministic special case of \glspl{mjls}, these results provide important insight into the inherent difficulty of verification problems in our setting.
In particular, they indicate that even in the absence of stochastic switching, reachability and model-checking problems quickly approach the frontier of decidability, thereby setting clear limits on what can be expected for the more general case of \glspl{mjls}.

\emph{\Acrfullpl{gmp}} extend finite-state \glspl{dtmc} to an uncountable state space by replacing the transition probability matrix with a stochastic kernel defined over a measurable state space.
Existing work explores a \gls{pctl} semantics for \glspl{gmp} and shows that reach-avoid properties can be characterized via Bellman fixed-point equations~\cite{ramponi2010connections}.
As we show in \autoref{eq:gmp-lifting}, \glspl{mjls} can be expressed in the framework of \glspl{gmp} and, thus, these results are directly relevant to our setting.
However, just like \glspl{mjls}, \glspl{gmp} should also be affected by the aforementioned difficulties in deciding reachability problems.

\section{Background}\label{sec:background}

\subsubsection*{Notation.}\label{sec:notation}
The indicator function that evaluates to \(1\) if \(x\in S\) and \(0\) otherwise is denoted \(\indicator{S}{x}\).
For \(i\in\set{1,\dots,m}\), let \(\basevect{i}\) be the \(i\)-th canonical base vector for the \(m\)-dimensional space.
Let \(\identmatrix{n}\) and \(\zeromatrix{n}\) be the identity and all-zeros matrix of size \(n\times n\), respectively.
\(A\transpose\) denotes the transpose of matrix \(A\).
Further, \(\spectral{A} = \max \abs{\lambda_{i}}\), where \(\lambda_{i}\) ranges over the eigenvalues of \(A\), denotes the spectral radius of matrix~\(A\), which is the maximum of the absolute values of \(A\)'s eigenvalues.

The block diagonal matrix \(\blkdiag{\statematrix{1},\dots, \statematrix{m}}\) is defined as the matrix that puts \(\statematrix{1}\) to \(\statematrix{m}\) on the main diagonal and fills all other entries with~\(0\).
We further define the Kronecker product \(A \kron B\) of matrices \(A\) and \(B\) as usual.

We will make use of a function \(\mvec{}\) that vectorizes an \(m\times n\) matrix to a vector of length~\(mn\) by stacking the matrix columns underneath each other, \ie \(\mvec{\strut(a_{ij})} = \begin{bsmallmatrix}a_{11} & \cdots & a_{m1} & \cdots & a_{1n} & \cdots & a_{mn}\end{bsmallmatrix}\transpose\).
Its inverse function \(\unmvec{m\times n}{}\) transforms a vector of length \(mn\) into an \(m\times n\) matrix.
The function \(\mvech{}\) carries this concept to symmetric matrices and only vectorizes the lower triangular elements.
In particular, \(\mvech{}\) ``half-vectorizes'' a symmetric \(n \times n\) matrix to a vector of length \(\frac{1}{2}n(n+1)\) and is given by \(\mvech{\strut(a_{ij})} = \begin{bsmallmatrix}a_{11} & \cdots & a_{n1} & a_{22} & \cdots & a_{n2} & \cdots & a_{nn}\end{bsmallmatrix}\transpose\).

Finally, we define a function \(\summn{mn}{n}{}\) that reshapes a vector of length~\(mn\) into an \(n\times m\) matrix (via \(\unmvec{n\times m}{}\)) and then sums its columns, resulting in a vector of \(n\)~entries.
Formally, it is defined as 
\(
  \summn{mn}{n}{\begin{bsmallmatrix}
      a_{1} & \cdots & a_{\mi{mn}}
    \end{bsmallmatrix}\transpose}
  = (1_{m}\transpose \kron \identmatrix{n}) \cdot \begin{bsmallmatrix}
      a_{1} & \cdots & a_{\mi{mn}}
    \end{bsmallmatrix}\transpose
\), where \(1_{m}\) is the all-ones vector of length~\(m\).

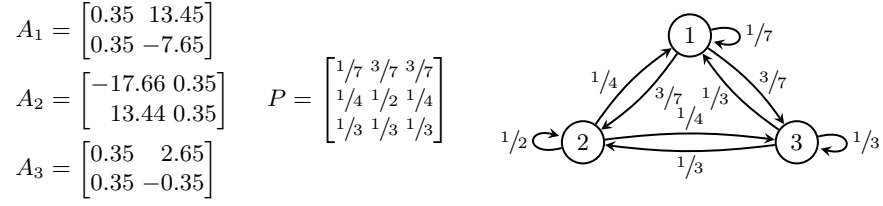
\begin{figure}[t]
  \centering%
  \begin{minipage}[c]{0.24\linewidth}
    \(\begin{aligned}[t]
      \statematrix{1} & = \begin{bmatrix}
                            0.35 & \phantom{-}\mathllap{1}3.45 \\ 0.35 & -7.65
                          \end{bmatrix}  \\
      \statematrix{2} & = \begin{bmatrix}
                            -17.66 & 0.35 \\ \phantom{-}13.44 & 0.35
                          \end{bmatrix} \\
      \statematrix{3} & = \begin{bmatrix}
                            0.35 & \phantom{-}2.65 \\ 0.35 & -0.35
                          \end{bmatrix}
    \end{aligned}\)%
  \end{minipage}\hspace{0.4cm}%
  \begin{minipage}[c]{0.20\linewidth}
    \(\probmatrix = \begin{bmatrix}
      \sfrac{1}{7} & \sfrac{3}{7} & \sfrac{3}{7} \\
      \sfrac{1}{4} & \sfrac{1}{2} & \sfrac{1}{4} \\
      \sfrac{1}{3} & \sfrac{1}{3} & \sfrac{1}{3}
    \end{bmatrix}%
    \)%
  \end{minipage}\hspace{0.1cm}%
    \begin{minipage}[c]{0.5\linewidth}
      \centering%
      \begin{tikzpicture}[baseline]
        \node[state, inner sep=3pt, minimum size=0pt] (m1) {1};
        \node[state, inner sep=3pt, minimum size=0pt] (m2) [below left=of m1] {2};
        \node[state, inner sep=3pt, minimum size=0pt] (m3) [below right=of m1] {3};
        
        \path[->]
        (m1)  edge [loop right] node [xshift=-1pt] {\(\sfrac{1}{7}\)} ()
        edge [bend left=10] node [right,xshift=1pt,yshift=-1pt] {\(\sfrac{3}{7}\)} (m2)
        edge [bend left=10] node [right] {\(\sfrac{3}{7}\)} (m3)
        (m2)  edge [bend left=10] node [left,xshift=-1pt] {\(\sfrac{1}{4}\)} (m1)
        edge [loop left] node [xshift=1pt] {\(\sfrac{1}{2}\)} ()
        edge [bend left=7] node [above,yshift=-2pt] {\(\sfrac{1}{4}\)} (m3)
        (m3)  edge [bend left=10] node [left,xshift=0.5pt,yshift=-1pt] {\(\sfrac{1}{3}\)} (m1)
        edge [bend left=7] node [below,yshift=2pt] {\(\sfrac{1}{3}\)} (m2)
        edge [loop right] node [xshift=-1pt] {\(\sfrac{1}{3}\)} ()
        ;
      \end{tikzpicture}%
  \end{minipage}%
  \caption{An \gls{mjls} \(\mathfrak{J} = (\probmatrix, \set{\statematrix{1}, \statematrix{2}, \statematrix{3}})\) with the \gls{dtmc} visualized on the right.}
  \label{fig:running-example}
\end{figure}

\subsubsection*{\Glspl{lds}.} A \acrfull{lds} is described by a state matrix~\(A \in \RR^{n\times n}\), and, given an initial state~\(x(0) \in \RR^{n}\), evolves according to the equation \(x(k+1) = Ax(k)\).
Properties of \glspl{lds}, in particular stability~\cite{brogan1985modern,chen1984linear} and reachability (here: orbit problem)~\cite{kannan1986orbit,chonev2016orbit}, have been studied extensively.

\subsubsection*{\Glspl{mjls}.} \acrfullpl{mjls} generalize \glspl{lds} by allowing for switching between a finite set of state matrices.
The switching is governed by an underlying \acrfull{dtmc}, whose modes are associated with the corresponding state matrices.
Formally, an \gls{mjls}~\(\mathfrak{J}\) with \(m\)~modes and \(n\)~continuous state variables is a tuple \((\probmatrix, \set{\statematrix{1},\ldots,\statematrix{m}})\), where \(\probmatrix\in\RR^{m\times m}\) is the transition probability matrix of Markov chain \(\set{\theta(k) \where k \in \NN}\) that assumes values from \(1\) to \(m\) and \(\statematrix{1}\) to \(\statematrix{m} \in \RR^{n \times n}\) are the state matrices.
\(\transprob{\theta}{\theta'}\) denotes the transition probability from mode~\(\theta\) to \(\theta'\), \ie \(\prob{\theta(k + 1) = \theta' \,\vert\, \theta(k) = \theta}\) for \(k \in \NN\).
Let \(x(k) \in \RR^{n}\) denote the continuous state at time \(k\in\NN\).
Then, this continuous state evolves according to the equation \(x(k + 1) = \statematrix{\theta(k)}x(k)\)~\cite{costa2005discrete}.

\begin{example}
\autoref{fig:running-example} shows an \gls{mjls} with three modes, corresponding state matrices \(\statematrix{1}\), \(\statematrix{2}\), and \(\statematrix{3}\), and two continuous state variables (values are rounded, see Appendix~\ref{sec:example-mjls} for the precise values).
\end{example}

\subsubsection*{\gls{mjls} Measurability.}The state space~\(\statespace\) of \gls{mjls}~\(\mathfrak{J}\) is of the form \(\RR^{n} \times \set{1, \ldots, m}\).
An infinite path of \(\mathfrak{J}\) is a sequence \((x_{0}, \theta_{0}) (x_{1}, \theta_{1}) (x_{2}, \theta_{2}) \cdots \in \statespace^{\omega}\), where for all \(i \geq 0\), \(\transprob{\theta_{i}}{\theta_{i+1}} > 0\) and \(x_{i+1} = \statematrix{\theta_{i}}x_{i}\).
Similarly, finite paths of length \(k \in \NN\) are sequences \((x_{0}, \theta_{0}) \cdots (x_{k}, \theta_{k}) \in \statespace^{k+1}\).
Given a state \(\xtheta \in \statespace\), \(\paths{\xtheta}\) and \(\paths[k]{\xtheta}\) describe the set of all infinite and finite paths of length~\(k\), respectively, beginning in \(\xtheta\).

The cylinder set~\(\cylset{\pi}\) of a finite path~\(\pi\) is the set of all infinite extensions of~\(\pi\).
The probability of such a cylinder set is defined via the corresponding cylinder set in the Markov chain~\cite{baier2008principles}:
\[
  \prob{\cylset{(x_{0}, \theta_{0}) \cdots (x_{k}, \theta_{k})}} = \prob{\cylset{\theta_{0} \cdots \theta_{k}}} = \prod_{{}i=0}^{k-1} \transprob{\theta_{i}}{\theta_{i+1}}
\]
The \(\sigma\)-algebra~\(\Sigma^{\mathfrak{J}}\) associated with \gls{mjls}~\(\mathfrak{J}\) is the smallest \(\sigma\)-algebra that contains all cylinder sets \(\cylset{\hat{\pi}}\) where \(\hat{\pi}\) ranges over all finite paths in \(\mathfrak{J}\).

\subsubsection*{\Glspl{gmp}.} A \acrfull{gmp}~\cite{ramponi2010connections} is a discrete-time stochastic process on a Borel set~\(X\), whose evolution is characterized by a stochastic kernel \(Q\colon X \times \mathfrak{B}(X) \to [0,1]\), where \(\mathfrak{B}(X)\) is the Borel \(\sigma\)-algebra on \(X\). 
\glspl{gmp} generalize \glspl{mjls}, as we can lift an $m$-dimensional \gls{mjls}~\(\mathfrak{J}\) to \glspl{gmp} by choosing \(X = \statespace\) and
\begin{equation}\label{eq:gmp-lifting}
  \Q{\xtheta}{B} = \sum_{\theta'=1}^{m} \indicator{B}{(\statematrix{\theta}x, \theta')} \cdot \transprob{\theta}{\theta'}
\end{equation}

\subsection{Stability of \Glspl{mjls}}

There are different notions of stability for \glspl{mjls}.
In the following, we summarize well-established results on \emph{mean} and \emph{mean-square stability}, which analyse the system state's expected value and uncentred variance, respectively, on the long run~\cite{lian2015mean,costa2005discrete}.
All results in the remainder of this section stem from~\cite{costa2005discrete}, notationally adjusted to our setting. 
The first and second (uncentred) moment at time \(k\in\NN\) of an \gls{mjls}~\(\mathfrak{J}\) with initial state \(\xtheta \in \statespace\) are defined as follows:
\begin{align*}
  \expectation{\xtheta}{x(k)}{} & = \sum_{\hspace{-15pt}\mathrlap{(x_{0}, \theta_{0}) \cdots (x_{k}, \theta_{k}) \in \paths[k]{\xtheta}}} x_{k} \cdot \prob{\cylset{\theta_{0} \cdots \theta_k}}                                                                         \\
  \variance{\xtheta}{x(k)}{}    & = \expectation{\xtheta}{x(k)x(k)\transpose{}}{} = \sum_{\hspace{-15pt}\mathrlap{(x_{0}, \theta_{0}) \cdots (x_{k}, \theta_{k}) \in \paths[k]{\xtheta}}} x_{k} x_{k}\transpose{} \cdot \prob{\cylset{\theta_{0} \cdots \theta_k}} \\
\end{align*}

Note that, as usual for multivariate random variables, the first moment is a vector of the first moments of each random variable, whereas the second (uncentred) moment is an \(n\times n\) symmetric matrix.
In this matrix, position \((i,j)\) represents the covariance between random variable \(x_{i}\) and \(x_{j}\), which are the respective continuous variables in our case.
Further note that while classically, variance is defined as the centred second moment \(\mathbb{E}\left[(X - \mathbb{E}\left[X\right])^2\right]\), this definition uses the uncentred moment.
However, \(\mathbb{E}\left[(X - \mathbb{E}\left[X\right])^2\right] = \mathbb{V}\left(X\right) - (\mathbb{E}\left[X\right])^2\).

\begin{definition}
  An \gls{mjls}~\(\mathfrak{J}\) is \emph{\gls{ms}} iff for \(k\to\infty\),
  \[
    \exists \mu\in\RR^n.\, \forall \xtheta\in\statespace.\, \expectation{\xtheta}{x(k)}{} \to \mu
  \]
  and it is \emph{\gls{mss}} iff for \(k\to\infty\),
  \[
    \exists v \in\RR^{n\times n}.\, \forall \xtheta\in\statespace.\, \variance{\xtheta}{x(k)}{} \to v
  \]
\end{definition}

Mean-square stability obviously implies mean stability.
Due to the quantifier ordering, both stability notions require the limit value to be the same for all initial states.
For our definition of \glspl{mjls}, this means that the only possible instantiation of \(\mu\) and \(v\) can be the constant zero vector and matrix, respectively.

To check an \gls{mjls}~\(\mathfrak{J}\) for stability, some further definitions are necessary.
For some initial state~\(\xtheta\in\statespace\) and Markov chain mode~\(\theta'\in\set{1,\dots,m}\),
\begin{align*}
  \expectation{\xtheta}{x(k)}{\theta(k) = \theta'} &= \sum_{\hspace{-15pt}\mathrlap{(x_{0}, \theta_{0})\cdots(x_{k}, \theta') \in \paths[k]{\xtheta}}} x_{k} \cdot \prob{\cylset{\theta_{0}\cdots \theta'}} \\
  \variance{\xtheta}{x(k)}{\theta(k) = \theta'} &= \sum_{\hspace{-15pt}\mathrlap{(x_{0}, \theta_{0})\cdots(x_{k}, \theta') \in \paths[k]{\xtheta}}} x_{k} x_{k}\transpose{} \cdot \prob{\cylset{\theta_{0}\cdots \theta'}}
\end{align*}
denote the \gls{mjls} state's expected value and uncentred variance at time~\(k\) with initial state~\(\xtheta\), restricted on the Markov chain mode~\(\theta'\) at time~\(k\).
Note that the sum of \(\expectation{\xtheta}{x(k)}{\theta(k) = i}\) over all \(i \in \set{1,\dots,m}\) is the overall expected state value at time~\(k\) starting in \(\xtheta\), \ie \(\expectation{\xtheta}{x(k)}{}\).
The analogous property for \(\sum_{i=1}^m \variance{\xtheta}{x(k)}{\theta(k) = i} = \variance{\xtheta}{x(k)}{}\) holds as well. Furthermore,  
 the following two equalities hold:
  \begin{align*}
    \expectation{\xtheta}{x(k+1)}{\theta(k+1) = j} & = \sum_{i=1}^{m} \transprob{i}{j} \cdot \statematrix{i} \cdot \expectation{\xtheta}{x(k)}{\theta(k) = i}                                \\
    \variance{\xtheta}{x(k+1)}{\theta(k+1) = j}    & = \sum_{i=1}^{m} \transprob{i}{j} \cdot \statematrix{i} \cdot \variance{\xtheta}{x(k)}{\theta(k) = i} \cdot \statematrix{i}\transpose{}
  \end{align*}
which can be exploited by an operator~\(\boperator_{\mathfrak{J}}\) that propagates the conditional expected value to the next time step, where
\(
  \boperator_{\mathfrak{J}} = (\probmatrix\transpose \kron \identmatrix{n}) \cdot \blkdiag{\statematrix{1},\dots,\statematrix{m}}
\).
This operator satisfies:
\[
  \begin{bmatrix}
    \expectation{\xtheta}{x(k+1)}{\theta(k+1) = 1} \\
    \vdots                                         \\
    \expectation{\xtheta}{x(k+1)}{\theta(k+1) = m}
  \end{bmatrix}  = \boperator_{\mathfrak{J}} \cdot
  \begin{bmatrix}
    \expectation{\xtheta}{x(k)}{\theta(k) = 1} \\
    \vdots                                     \\
    \expectation{\xtheta}{x(k)}{\theta(k) = m}
  \end{bmatrix}
\]
Similarly, the operator
\(
  \toperator_{\mathfrak{J}} = (\probmatrix\transpose\kron \identmatrix{n^2}) \cdot \blkdiag{\statematrix{1} \kron \statematrix{1},\ldots, \statematrix{m}\kron \statematrix{m}}
\) propagates the conditional uncentred variance to the next time step.

Since both operators characterize the stepwise evolution of the respective moments, it is their characteristics that determine the respective limiting behaviour. 
Thus, we can resort to eigenvalue analysis to determine stability with respect to mean and uncentred variance: 

\begin{reptheorem}{thm:ms-and-mss-spectral-radius}[\cite{costa2005discrete}]
  An \gls{mjls}~\(\mathfrak{J}\) is \gls{ms} iff \(\spectral{\boperator_{\mathfrak{J}}} < 1\).
  It is \gls{mss} iff \(\spectral{\toperator_{\mathfrak{J}}} < 1\).
\end{reptheorem}

\begin{example}
  The \gls{mjls} from \autoref{fig:running-example} is neither  \gls{mss} nor \gls{ms}, because \({\spectral{\boperator_{\mathfrak{J}}} \approx 11.4}\).
\end{example}

\section{\glsentrytitlecase{pctl}{long} for \gls{mjls}}\label{sec:pctl}

Originally established to specify properties on probabilistic systems such as Markov chains~\cite{hansson1994logic}, \gls{pctl} is a temporal logic that suits well our modelling framework of \glspl{mjls}.
We define its syntax as usual:
\begin{align*}
  \Phi    & \Coloneqq \mi{true} \bnfmid a \bnfmid \Phi_1 \land \Phi_2 \bnfmid \neg \Phi \bnfmid \PCTLprob{J}{\varphi} \\
  \varphi & \Coloneqq \PCTLnext{\Phi} \bnfmid \PCTLbuntil{\Phi_1}{\Phi_2}{k} \bnfmid \PCTLuntil{\Phi_1}{\Phi_2}
\end{align*}
where \(a \in \ap_{x} \sqcup \ap_{\theta}\), \(J \subseteq [0, 1]\), and \(k \in \NN\).
We denote the step-bounded fragment of \gls{pctl} with \gls{pctlm}, \ie all constructs except the unbounded until.
We define \(\ap_{x}\) and \(\ap_{\theta}\) as the set of atomic propositions on the continuous state and the discrete state, respectively, with \(\sqcup\) denoting disjoint union.
States \(\xtheta \in \statespace\) can be labelled with atomic propositions:
\[
  L(\xtheta) = \set{a \in \ap_{x} \where x \in \Pi_a} \sqcup \set{a \in \ap_{\theta} \where a \in L_{\theta}(\theta)}
\]
where \(\Pi_a\) are convex polytopes, which are bounded convex subsets of \(\RR^n\) that can be represented equivalently either as the intersection of finitely many half-spaces or as the convex hull of finitely many vertices~\cite{ziegler1993polytopes}.
We restrict the labelling of the continuous state with atomic propositions to convex polytopes, since this class of sets provides a suitable balance between expressiveness and computational tractability.
In particular, membership of a state in a convex polytope can be checked efficiently via its halfspace representation~\cite{groetschel1993geometric}.
Further, the model-checking algorithm presented below requires repeated set operations and transformations.
Convex polytopes are well suited as a basis for this, since operations such as affine (and, thus, linear) transformations and intersection preserve the polytopic structure~\cite{ziegler1993polytopes}.
Thus, the sets arising during the recursive evaluation of formulas remain computationally manageable.
By contrast, more general classes of sets may fail to retain such a tractable representation under these operations, making them less suitable for algorithmic verification.

Let \(\mathfrak{J}\) be an \gls{mjls}, \(a \in \ap_{x} \sqcup \ap_{\theta}\), and \(\xtheta \in \statespace\).
We define the satisfaction relation \(\vDash\) of \gls{pctl} state formulas analogously to Markov chains~\cite{baier2008principles}:
\begin{align*}
   & \xtheta  \vDash \mi{true}            &  &                                                                                       &  & \xtheta  \vDash a                   &  & \text{iff } a \in L(\xtheta)                                  \\
   & \xtheta \vDash \neg \Phi             &  & \text{iff } \xtheta \nvDash \Phi                                                                        &  & \xtheta  \vDash \Phi_1 \land \Phi_2 &  & \text{iff } \xtheta \vDash \Phi_1 \land \xtheta \vDash \Phi_2 \\
   & \xtheta \vDash \PCTLprob{J}{\varphi} &  & \text{iff } \mathrlap{\prob{\set{\pi \in \paths{\xtheta} \where \pi \vDash \varphi}} \in J}\hspace{2cm}
\end{align*}
Similarly, we define for an infinite path \(\pi\in\statespace^{\omega}\):
\begin{align*}
  \pi & \vDash \PCTLnext{\Phi}                &  & \text{iff } \pi[1] \vDash \Phi                                                                                 \\
  \pi & \vDash \PCTLbuntil{\Phi_1}{\Phi_2}{k} &  & \text{iff } \exists j \leq k.\, \pi[j] \vDash \Phi_2 \land \forall i < j.\, \pi[i] \vDash \Phi_1 \\
  \pi & \vDash \PCTLuntil{\Phi_1}{\Phi_2}     &  & \text{iff } \exists j \geq 0.\, \pi[j] \vDash \Phi_2 \land \forall i < j.\, \pi[i] \vDash \Phi_1
\end{align*}
where for a path \(\pi = (x_0, \theta_0)(x_1, \theta_1)\cdots\) and \(i \geq 0\), \(\pi[i]\) denotes state \((x_i, \theta_i)\) of~\(\pi\).

\begin{reptheorem}{thm:pctl-measurability}
  For each \gls{pctl} path formula~\(\varphi\) and state \(\xtheta \in \statespace\) of an \gls{mjls}~\(\mathfrak{J}\), the set \(\set{\pi \in \paths{\xtheta} \where \pi \vDash \varphi}\) is measurable, \ie it is an element of \(\Sigma^{\mathfrak{J}}\).
\end{reptheorem}

So, the probabilities of the relevant sets can be constructed from the available operations on the cylinder sets.

\subsubsection{Relation to \gls{gmp}.}

As per the discussion above, \glspl{mjls} can be embedded into \acrfullpl{gmp}~\cite{ramponi2010connections} with the lifting of \autoref{eq:gmp-lifting}.
This is consistent with the respective definitions of \gls{pctl}:
\begin{replemma}{lem:gmp-equiv-algo}
  Let Borel set \(X = \statespace\) with the associated Borel \(\sigma\)-algebra and stochastic kernel~\(Q\) from \autoref{eq:gmp-lifting}.
  Let \(\vDashGMP\) denote the \gls{pctl} satisfaction relation for \glspl{gmp} as given in~\cite{ramponi2010connections}.
  Then, for any \gls{pctl} formula~\(\Phi\) and state~\(\xtheta \in \statespace\), \(\xtheta \vDashGMP \Phi\) iff~\({\xtheta \vDash \Phi}\).
\end{replemma}

\subsubsection{Model Checking \Acrshort{pctlm}.}

Just like \gls{pctl} model checking on Markov chains (or CTL on Kripke structures), its model checking on \glspl{mjls} is based on a bottom-up traversal of the parse tree of the formula in question.
Along the way, the algorithm computes the satisfaction set \(\sat{\Phi} = \set{\xtheta \in \statespace \where \xtheta \vDash \Phi}\) of the respective subformula~\(\Phi\) as described in the following, where we first restrict to \gls{pctlm} for the sake of exposition.

\begin{reptheorem}{thm:correctness-mc-boundedpctl}
  Let \(\mathfrak{J}\) be an \gls{mjls} and \(\Phi\) a \gls{pctlm} formula.
  For every \(\xtheta \in \statespace\), it holds that:
  \begin{align*}
    &\sat{\ttrue}= \statespace && \sat{a}                    = 
    \begin{cases}
      \Pi_{a}\times \set{1,\ldots,m} &\text{if } a \in \ap_{x}\\
      \RR^{n} \times L_{\theta}(\theta) &\text{if } a \in \ap_{\theta}
    \end{cases}   \\
    &\sat{\neg \Phi} = \statespace \setminus \sat{\Phi} && \sat{\Phi_1 \land \Phi_2}  = \sat{\Phi_1} \cap \sat{\Phi_2}                        \\
    &\sat{\PCTLprob{J}{\PCTLnext{\Phi}}}                    = \mathrlap{\set{\xtheta \in \statespace \where \sum_{\theta' = 1}^{m} \indicator{\sat{\Phi}}{(\statematrix{\theta}x, \theta')} \cdot \transprob{\theta}{\theta'} \in J}}                                                                                         \\
    &\sat{\PCTLprob{J}{\PCTLbuntil{\Phi_{1}}{\Phi_{2}}{k}}}  = \mathrlap{\set{\xtheta \in \statespace \where U_{k}\left(\xtheta\right) \in J}}
  \end{align*}
  where
  \begin{align*}
    U_{0}\left(\xtheta\right)
     & = \indicator{\sat{\Phi_{2}}}{\xtheta} \\
    U_{k+1}\left(\xtheta\right)
     & =
    U_{0}\left(\xtheta\right)
    + \indicator{\sat{\Phi_{1}} \setminus \sat{\Phi_{2}}}{\xtheta} \sum_{\theta' = 1}^{m} \transprob{\theta}{\theta'} \cdot U_{k}\left((\statematrix{\theta}x,\theta')\right)
  \end{align*}
\end{reptheorem}

\subsubsection{Model Checking \(\PCTLuntil{\Phi_1}{\Phi_2}\).}

We now turn to the model checking of the full logic.
The only missing case is the computation of \(\sat{\PCTLuntil{\Phi_1}{\Phi_2}}\), where we can assume that $\sat{\Phi_1}$ and $\sat{\Phi_2}$ are already computed. 
In principle, the work of~\cite{ramponi2010connections} provides a guideline, since it establishes a connection between \gls{pctl} and dynamic programming (albeit not formulated for satisfaction sets, but for individual states) for the \gls{gmp} setting. 
More concretely, it views the until~\(\mathcal{U}\) as the limit of the bounded-until~\(\mathcal{U}^{\leq k}\) for \(k \to \infty\), which is quite natural.
However, the existence of this limit appears not to be guaranteed. 
The reason is as follows: 
Since \glspl{lds} are a special case of \glspl{mjls} (with only a single mode and \(\probmatrix = \begin{bsmallmatrix}1\end{bsmallmatrix}\)), \glspl{lds} are also a special case of \glspl{gmp}. 
As an \gls{lds} only exhibits a single path given an initial state, reachability in the path-based setting (such as \gls{ltl} and \acrshort{mso}) coincides with reachability in the branching-time view, such as in \gls{pctl}.

Thus, the negative results with respect to decidability for \glspl{lds} reachability carry over to both \glspl{mjls} and \glspl{gmp}.
To the best of our knowledge, this implication for \glspl{gmp} has not been observed before.
Notably: if \(\sat{\Phi_2}\) (\ie the space that is supposed to be reached eventually) is at least 4-dimensional, \gls{lds} reachability is an open problem~\cite{chonev2016orbit} linked to the Skolem problem.
(On the other hand, point-to-point reachability is decidable for \glspl{lds} in polynomial time, and so is point-to-hyperplane reachability for system dimensions up to~\(4\)~\cite{karimov2022decidable}.)

This contrast is striking.
While the fixed-point characterization for \glspl{gmp}~\cite{ramponi2010connections} suggests a conceptually simple recursive computation, even basic reachability questions for the strictly simpler class of \glspl{lds} are unresolved in general.
This indicates that the apparent simplicity of the recursion does not translate into algorithmic tractability. 
We conclude this section with the observation that the model-checking problem for \gls{pctl} beyond \gls{pctlm} on \glspl{mjls} is open. 
\section{Adding Stability Properties to \gls{pctl}}\label{sec:expectation}

Our main motivation, as stated earlier, lies in the fact that a negative answer to the stability analysis of an \gls{mjls} may be overly conservative.
Two scenarios immediately come to mind in which a more refined answer may be useful:
\begin{enumerate}
  \item Even though the moment converges for all initial states, it does not converge to the same value (\ie the origin) for all initial states.
        Thus, the \gls{mjls} is not considered stable.
        However, in this scenario, convergence itself might be considered of higher importance than convergence to the origin.
  \item For some set of initial states, the moment diverges.
        Again, the \gls{mjls} is not considered stable.
        However, these initial states may be insignificant, \eg because they correspond to impossible configurations in the real-world process.
        In this case, one would like to restrict the stability analysis to the set of relevant initial states.
\end{enumerate}

\begin{example}
  A combination of the two scenarios above is well possible.
  Recall that the system illustrated in \autoref{fig:running-example} is not \gls{ms}.
  In fact, for most initial states, the first moment diverges.
  However, there exist initial states whose first moment stabilizes in the long run (though not all to a common value).
  For instance, the long-run expected value of initial state \(\begin{bsmallmatrix}16.047 & -0.254\end{bsmallmatrix}\transpose\) with \(\theta = 1\) is \(\begin{bsmallmatrix}17.377 & 0.8\end{bsmallmatrix}\transpose\), while for \(\begin{bsmallmatrix}0.324 & 31.814\end{bsmallmatrix}\transpose\) with \(\theta = 2\) it is \(\begin{bsmallmatrix}-2.69 & -0.124\end{bsmallmatrix}\transpose\).
\end{example}
These observations motivate a logical framework capable of specifying convergence properties relative to sets of initial states of interest.
We address this with an extension of \gls{pctl} for \glspl{mjls} to express properties on the system's first and second moments, given a specific initial state.
To arrive there, we extend the atomic propositions of \gls{pctl} with corresponding propositions in both step-bounded and step-unbounded variants.
Just like the original atomic propositions, these are state formulas and can be nested within other \gls{pctl} formulas.
The converse, \ie nesting other \gls{pctl} (state) formulas within the new operators as it is possible for \eg \(\mathcal{X}\), is not supported in our choice of extension.
Other than that, the semantics as defined in \autoref{sec:pctl} remains unchanged.
We stipulate 
\[
  a \in \mi{AP}_{x} \sqcup \mi{AP}_{\theta} \sqcup \set{\PCTLbexp{k}{\Pi}, \PCTLexp{\Pi}, \PCTLbvar{k}{\Xi}, \PCTLvar{\Xi} \where k\in\NN}
\]
where \(\Pi\) and \(\Xi\) are convex polytopes in \(\RR^{n}\) and \(\RR^{\frac{1}{2}n(n+1)}\), respectively.
For a given convex polytope~\(\Pi\), \(\PCTLbexp{k}{\Pi}\) is meant to label all states in which the first moment up to step \(k\) on average lies in~\(\Pi\), and \(\PCTLexp{\Pi}\) analogously all states in which the corresponding long-run expected value lies in \(\Pi\).
The operators \(\PCTLbvar{k}{\Xi}\) and \(\PCTLvar{\Xi}\) are defined in the same way, but for the second moment.
We restrict~\(\Pi\) and~\(\Xi\) to convex polytopes again for the possibility of efficient representation, while still providing an intuitive modelling framework.
Since the second moment is an \(n \times n\) matrix, the bounding polytope~\(\Xi\) has to reflect this and put bounds on \(n^2\) variables.
Its dimensionality can be reduced to \(\frac{1}{2}n(n+1)\) since the second moment matrix is symmetric.
We use the \(\mvech{}\) operator (see \autoref{sec:notation}) to stack the lower triangular part of a matrix into a vector.
With this, we define:
\begin{align*}
   & \PCTLbexp{k}{\Pi} \in L(\xtheta) \quad\text{ iff }\quad
  \frac{1}{k+1}\sum_{i=0}^{k} \, \expectation{\xtheta}{x(i)}{} \quad\in \Pi                         \\
   & \PCTLexp{\Pi} \in L(\xtheta) \quad\text{ iff }\quad
  \lim_{k \to \infty} \frac{1}{k+1} \sum_{i = 0}^{k} \, \expectation{\xtheta}{x(i)}{} \quad\in  \Pi \\
   & \PCTLbvar{k}{\Xi} \in L(\xtheta) \quad\text{ iff }\quad
  \mvech{\frac{1}{k+1}\sum_{i=0}^{k} \, \variance{\xtheta}{x(i)}{}} \quad\in \Xi                    \\
   & \PCTLvar{\Xi} \in L(\xtheta) \quad\text{ iff } \quad
  \mvech{\lim_{k \to \infty} \frac{1}{k+1} \sum_{i = 0}^{k} \, \variance{\xtheta}{x(i)}{}} \quad\in \Xi
\end{align*}

The connection to mean and mean-square stability is immediate:
in an \gls{mjls} that is \gls{ms} or \gls{mss}, respectively, the entire state space is labelled with \(\PCTLexp{\Pi}\) or \(\PCTLvar{\Xi}\) for any polytope~\(\Pi\) or \(\Xi\) that contains the origin.
However, the converse implication does not hold.
This is due to the use of time averages (Cesàro limits):
even if the time-averaged first or second moment converges, the corresponding moment sequence itself may fail to converge.
In particular, oscillating moment sequences may admit a well-defined Cesàro limit while lacking a pointwise limit.
This choice is consistent with the notion of limiting distributions in \glspl{dtmc}~\cite{kemeny1976finiteMCs,andova2004rewards} and contrasts possible alternative characterizations using limit superior and inferior (\cf~\cite{baier2025multiplicative}).
Hence, the proposed operators capture a slightly different notion of convergence than classical \gls{ms} and \gls{mss}, focusing on long-run average behaviour rather than pointwise limits, and remaining well-defined even in the presence of oscillations.
Further, note that variance is typically defined as the centred second moment.
However, using the uncentred second moment aligns more closely with classical mean-square stability.
Centring can be incorporated directly into the specification of the polytope~\(\Xi\) if desired.

Although the step-unbounded operators describe limit behaviour, and thus finite path prefixes do not affect the limit, nesting them within \gls{pctl} formulas is still meaningful.
Even a single transition may change the system dynamics to the extent that a state not satisfying $\mathcal{E}_{\Pi}$ (or $\mathcal{V}_{\Xi}$)  transitions to one that does, implying that \(\PCTLprob{J}{\PCTLnext{\PCTLexp{\Pi}}}\) for some interval~\(J\) can still be satisfied.
\begin{wrapfigure}[16]{r}{0.37\linewidth}
  \vspace{-5pt}
  \centering%
  \input{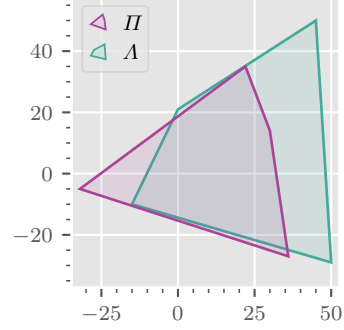}%
  \caption{A visualization of the polytopes from \autoref{ex:E-formulas}.}%
  \label{fig:polytopes}%
\end{wrapfigure}

\begin{example}\label{ex:E-formulas}
  Consider again the \gls{mjls} given in \autoref{fig:running-example}.
  Let \(\Pi \in \RR^2\) be the convex polytope defined by the vertices \((30, 14), (22, 35), (-32, -5)\), and \((36, -27)\).
  We are interested in the set of states that satisfy \(\PCTLexp{\Pi}\).
  Furthermore, we would like to allow for a five-step grace period to reach a state with the desired long-run expected value, while still avoiding unsafe states.
  We define another convex polytope~\(\Lambda\) with vertices \((-15, -10), (50, -29), (45, 50)\), and \((0, 21)\) of safe states.
  With a slight abuse of notation, atomic proposition~\(\Lambda\) labels all states in polytope~\(\Lambda\).
  Requiring the reach-avoid probability to be at least \(0.3\), we desire to check the PCTL formula \(\PCTLprob{\geq 0.3}{\PCTLbuntil{\Lambda}{\PCTLexp{\Pi}}{5}}\).
  For an illustration of the polytopes, see \autoref{fig:polytopes}.
  We shall return to the evaluation of the formulas in \autoref{ex:sat-E-formulas}.
\end{example}

In the following, we present the model-checking algorithms for \(\PCTLbexp{k}{\Pi}\) and \(\PCTLexp{\Pi}\).
The treatment of the variance operators follows analogously and will be discussed afterwards.
For the remainder of this section, let \(\mathfrak{J}\) be an \gls{mjls} with \(m\) modes and \(n\) continuous variables.

\subsubsection*{\(\sat{\PCTLbexp{k}{\Pi}}\).}

The satisfaction set of \(\PCTLbexp{k}{\Pi}\) for an \gls{mjls}~\(\mathfrak{J}\) relies on propagating the operator~\(\boperator_{\mathfrak{J}}\) for \(k\)~steps.
Since~\(\boperator_{\mathfrak{J}}\) propagates the conditional expected value, we initialize with the corresponding values at time~\(0\), \ie for a given state~\(\xtheta \in \statespace\):
\[
  \begin{bmatrix}\expectation{\xtheta}{x(0)}{\theta(0) = 1} & \cdots & \expectation{\xtheta}{x(0)}{\theta(0) = \theta} & \cdots & \expectation{\xtheta}{x(0)}{\theta(0) = m}\end{bmatrix}\transpose{}
\]
which reduces to \(\begin{bmatrix}0 & \cdots & 0 & x & 0 & \cdots & 0\end{bmatrix}\transpose{} = \basevect{\theta} \kron x\).
Finally, observe that:
\begin{equation}\label{eq:exp-via-bj}
  \expectation{\xtheta}{x(k)}{} = \summn{mn}{n}{\boperator_{\mathfrak{J}}^{k}(\basevect{\theta}\kron x)}
\end{equation}
The satisfaction set is then obtained as follows:

\begin{theorem}\label{thm:correctness-mc-boundedexp}
  Let \(\mathfrak{J}\) be an \gls{mjls} with \(n\) state variables, let \(k\in\NN\) and \(\Pi\) be a convex polytope in~\(\RR^{n}\).
  For every \(\xtheta \in \statespace\), it holds that:
  \[
    \sat{\PCTLbexp{k}{\Pi}} = \set{\xtheta \in \statespace \where \summn{mn}{n}{\left(\frac{1}{k+1}\sum_{i=1}^{k}\boperator_{\mathfrak{J}}^{i}\right)(\basevect{\theta} \kron x)} \in \Pi}
  \]
\end{theorem}

While the computation of the satisfaction seems to iterate over the uncountable state space, this is actually not necessary.
First, note that propagating \(\boperator\) and computing the Cesàro average is independent of the initial states.
The remainder is then computed for each mode of the Markov chain separately (and forming the union of the resulting sets):
Given a mode~\(\theta\), the multiplication of the Cesàro average and \((\basevect{\theta} \kron x)\) amounts to \(mn\) linear expressions over only \(n\) variables.
\(\summn{mn}{n}{}\) and the constraint of membership in~\(\Pi\) reduces this to an \(n\)-dimensional equation system of \(n\) variables.
The correctness for this part of the model-checking algorithm follows immediately from the above derivation.

Again, since the algorithm is based on solving linear equation systems and computing set membership of a convex polytope, the resulting satisfaction set remains representable.

\subsubsection*{Computing Long-run Cesàro Limit.}
First, we describe how to compute the Cesàro-averaged long-term first moment for one specific initial state~\(\xtheta\in\statespace\), \ie
\(
\lim_{k \to \infty} \frac{1}{k+1} \sum_{i = 0}^k \expectation{\xtheta}{x(i)}{}
\). We then turn to the computation of $\sat{\PCTLexp{\Pi}}$.

Obviously, if the initial state vector is the zero vector or the \gls{mjls} is \gls{ms}, the limit trivially is zero.
Otherwise, the analysis relies on linear-algebraic computations revolving around \autoref{eq:exp-via-bj}.
However, the approach differs in technical aspects if \(\boperator_{\mathfrak{J}}\) is not diagonalizable.

If \(\boperator_{\mathfrak{J}}\) is diagonalizable, it is possible to represent the vector~\(\basevect{\theta}\kron x\) of conditional expected values at time~\(0\) as a linear combination of the (complete set of potentially complex and linearly-independent) eigenvectors \(v_{1}\),\ldots, \(v_{mn}\) of~\(\boperator_{\mathfrak{J}}\), \ie \(\basevect{\theta} \kron x = \sum_{j = 1}^{mn} a_{j}v_{j}\).
Since \(\boperator_{\mathfrak{J}} v_{j} = \lambda_{j} v_{j}\) for all eigenpairs \((\lambda_{j}, v_{j})\), we get:
\begin{align*}
  \lim_{k \to \infty} \frac{1}{k+1} \sum_{i = 0}^k \expectation{\xtheta}{x(i)}{}
      & = \lim_{k \to \infty} \frac{1}{k+1}\sum_{i=0}^{k} \summn{mn}{n}{\boperator_{\mathfrak{J}}^i (\basevect{\theta} \kron x)} \\
   & = \lim_{k \to \infty} \frac{1}{k+1}\sum_{i=0}^{k} \summn{mn}{n}{\sum_{\smash{j} = 1}^{mn} a_{j}\lambda_{j}^{i}v_{j}}
\end{align*}

For those eigenvalues~\(\lambda_{j}\) that lie on or within the unit circle, \ie \(\abs{\lambda_{j}} \leq 1\), the limit can be determined per summand.
It is \(\summn{mn}{n}{a_{j}v_{j}}\) for \(\lambda_{j} = 1\) and \(0\) otherwise.
If there are summands with~\(\abs{\lambda_{j}} > 1\), the limit only exists if those summands are cancelled out, \ie \(\sum_{j | \lambda_j = \lambda} \summn{mn}{n}{a_{j}v_{j}} = 0\) for every \(\abs{\lambda} > 1\).

If \(\boperator_{\mathfrak{J}}\) is not diagonalizable, it is not possible to perform the above computations as there is no complete set of eigenvectors.
Instead, we use the Jordan normal form of \(\boperator_{\mathfrak{J}} = M \cdot \blkdiag{J_{m_{1}}(\lambda_{1}), \ldots, J_{m_{b}}(\lambda_{b})} \cdot M^{-1}\), where \(\lambda_{j}\) are eigenvalues with multiplicity~\(m_{j}\) and \(J_{m_{j}}(\lambda_{j})\) are the corresponding Jordan blocks.
A Jordan block \(J_{m_{j}}(\lambda_{j})\) is an \(m_{j} \times m_{j}\) square matrix with all-\(\lambda_{j}\) on its main diagonal and all-ones on the superdiagonal.
Substituting the Jordan normal form into the limit to be computed, we obtain:
\begin{align}
      & \lim_{k \to \infty} \frac{1}{k+1}\sum_{i=0}^{k} \summn{mn}{n}{\boperator_{\mathfrak{J}}^i (\basevect{\theta} \kron x)} = \label{eq:limit-jordan}                                                                     \\
   & \lim_{k \to \infty} \frac{1}{k+1}\sum_{i=0}^{k} \summn{mn}{n}{M\left(\blkdiag{J_{m_{1}}^{i}(\lambda_{1}),\ldots,J_{m_{b}}^{i}(\lambda_{b})}\right)M^{-1} (\basevect{\theta} \kron x)} \nonumber
\end{align}

For some \(i \geq m\), when zooming into Jordan block~\(J_{m}^{i}(\lambda)\), the \(k\)-th diagonal (\(k\geq 0\)) is constant and equal to \(\binom{i}{k}\lambda^{i-k}\).
In particular, the main diagonal (\ie 0-th diagonal) is equal to \(\lambda^{i}\).
The Cesàro limit of the above term can then be computed for every Jordan block separately and depends on the corresponding eigenvalue~\(\lambda\).
For eigenvalues~\(\lambda_{j}\) inside the unit circle, the limits of all entries of the Jordan block are \(0\).
For eigenvalues~\(\lambda_{j}\) on the unit circle, the limits on the main diagonal exist (\(1\) for \(\lambda = 1\) and \(0\) otherwise) but not for the superdiagonals.
For eigenvalues~\(\lambda_{j}\) outside the unit circle, all limits of the upper triangular matrix do not exist.
We compute the limit of each block symbolically as follows:
For a Jordan block~\(J_{m_{j}(\lambda_{j})}\) where \(\abs{\lambda_{j}} < 1\), we simply create a zero matrix.
If \(\abs{\lambda_{j}} \geq 1\), we create a matrix~\(K_{j}\) and fill the diagonals corresponding to non-converging entries with symbolic values~\(\infty_{ja_{j}}\).
Each non-converging diagonal is filled with a different symbolic value, denoted by the second index of~\(\infty_{ja_{j}}\), because the different non-converging sequences cannot cancel each other out.
Finally, we combine all \(K_{j}\) to a block matrix and perform the outer multiplications of \autoref{eq:limit-jordan}.
If all symbolic values vanish, the limit can be computed by applying \(\summn{mn}{n}{}\) to the multiplication result.

\subsubsection*{\(\sat{\PCTLexp{\Pi}}\).}

\begin{algorithm}[t]
  \caption{Satisfaction set of \(\PCTLexp{\Pi}\)}
  \label{alg:sat-long-term-exp}
  \begin{algorithmic}[1]
    \Procedure{SatE}{\(\Pi\)}
    \State \(\boperator_{\mathfrak{J}} \gets (\probmatrix\transpose \kron \identmatrix{n})\cdot\Call{BlkDiag}{\statematrix{1}, \ldots, \statematrix{m}}\)
    \If{\(\spectral{\boperator_{\mathfrak{J}}} < 1\)} \Comment{if \gls{ms}, convergence to origin for every state}
    \State \algorithmicif\ \(\begin{bsmallmatrix}0 & \cdots & 0\end{bsmallmatrix}\transpose\in \Pi\)
    \algorithmicthen\ \Return \(\statespace\)
    \algorithmicelse\ \Return \(\emptyset\)
    \ElsIf{\(\boperator_{\mathfrak{J}}\) is diagonalizable}
    \State \((\lambda_{1}, v_{1})\), \ldots, \((\lambda_{\mi{mn}}, v_{\mi{mn}}) \gets \Call{Eigenpairs}{\boperator_{\mathfrak{J}}}\)
    \LComment{move to coefficient domain}
    \State \(U \gets \set{(a_{1}, \ldots, a_{\mi{mn}}) \in \RR^{\mi{mn}} \where \exists \xtheta \in \statespace.\, \smash{\basevect{\theta}} \kron x = \sum_{j=1}^{mn} a_{j}v_{j}}\)
    \LComment{exclude all that diverge}
        \State \(V \gets U \cap \set{(a_{1}, \ldots, a_{\mi{mn}}) \in \RR^{\mi{mn}} \where \forall \abs{\lambda} > 1.\, \sum_{j | \lambda_{j} = \lambda} \summn{mn}{n}{a_{j}v_{j}} = 0}\)\label{line:set-to-0}
    \LComment{exclude all with result \(\notin \Pi\)}
    \State \(W \gets V \cap \set{(a_{1}, \ldots, a_{\mi{mn}}) \in \RR^{\mi{mn}} \where \sum_{j | \lambda_{j} = 1} \summn{mn}{n}{a_{j}v_{j}} \in \Pi}\)\label{line:element-H}
    \LComment{back to state domain}
    \State \Return \(\set{\xtheta \in \statespace \where \exists (a_{1}, \ldots, a_{\mi{mn}}) \in W.\, \basevect{\theta} \kron x = \sum_{j=1}^{mn} a_{j}v_{j}}\)
    \Else
    \State \(M\), \((\lambda_{1}, a_{1}), \ldots, (\lambda_{\ell}, a_{\ell}) \gets \Call{JordanDecomp}{\boperator_{\mathfrak{J}}}\)
    \ForAll{\(j \gets 1\) to \(\ell\)}
    \State \(K_{j} \gets \zeromatrix{a_{j}}\) \Comment{start off with zero matrix}
    \LComment{\textsc{FillDiags} fills \(j\)-th diag.\@ with \(j\)-th arg., starting with main diag.}
    \If{\(\abs{\lambda_{j}} > 1\)}
    \(K_{j} \gets \Call{FillDiags}{\infty_{j1}, \infty_{j2}, \ldots, \infty_{ja_{j}}}\)
    \ElsIf{\(\lambda_{j} = 1\)}
    \(K_{j} \gets \Call{FillDiags}{1, \infty_{j2}, \ldots, \infty_{ja_{j}}}\)
    \ElsIf{\(\abs{\lambda_{j}} = 1\)}
    \(K_{j} \gets \Call{FillDiags}{0, \infty_{j2}, \ldots, \infty_{ja_{j}}}\)
    \EndIf
    \EndFor
    \State \(\mi{r} \gets \set{(x, \theta) \in U \where \summn{mn}{n}{M \cdot \Call{BlkDiag}{K_{1},\ldots, K_{\ell}} \cdot M^{-1} (\basevect{\theta} \kron x)} \in \Pi}\)
    \State \Return \(\mi{r}\) \Comment{\(\in \Pi\) implies that no \(\infty\) is contained in result}
    \EndIf
    \EndProcedure
  \end{algorithmic}
\end{algorithm}

\autoref{alg:sat-long-term-exp} computes \(\sat{\PCTLexp{\Pi}}\) in a fashion similar to what the above algorithm does for a single initial state.
However, starting with the entire state space, this algorithm removes all initial states that do not satisfy the formula step by step via set intersections.

The first step is again to check for mean stability.
If the \gls{mjls} is \gls{ms}, \(\sat{\PCTLexp{\Pi}}\) encompasses the entire state space provided the polytope~\(\Pi\) contains the origin and evaluates to the empty set otherwise.
If the system is not \gls{ms}, we again distinguish between diagonalizable and non-diagonalizable \(\boperator_{\mathfrak{J}}\).

For diagonalizable~\(\boperator_{\mathfrak{J}}\), we proceed via linear combinations of the eigenvectors of~\(\boperator_{\mathfrak{J}}\).
In particular, we first compute the set of coefficients \(a_{1}\) to \(a_{mn}\) for which the result of \(\sum_{j=1}^{mn} a_{j}v_{j}\) is a vector of shape \(\basevect{\theta} \kron x\) for some \(\xtheta\in\statespace\).
In other words, we want to compute the set of coefficients such that all rows but \(n(\theta-1)+1\) to \(n(\theta-1)+n\) are set to zero, for every mode~\(\theta\).
Algorithmically, this can be realized by computing the nullspace of the matrix obtained by removing those rows from the modal matrix of \(\boperator_{\mathfrak{J}}\).
Then, from the resulting set, filter out all initial states whose first moment diverges, \ie all those coefficients where \(\abs{\lambda} > 1\) and the contribution of this eigenvalue is not cancelled out.
Finally, we remove all elements for which the limit value lies outside the desired polytope~\(\Pi\) and translate the resulting set back to the state domain.

If \(\boperator_{\mathfrak{J}}\) is not diagonalizable, we compute the same symbolic matrices~\(K_{i}\) as above, since they are independent of the initial state.
Then, compute the set of all initial states for which the multiplication with this Jordan decomposition does not contain any of the symbolic values as those represent divergence and return all those initial states that are included in polytope~\(\Pi\).

\begin{theorem}\label{thm:correctness-mc-exp}
  Let \(\mathfrak{J}\) be an \gls{mjls} with \(n\) state variables, let \(\Pi\) be some convex polytope in~\(\RR^n\).
  For every \(\xtheta \in \statespace\), it holds that
  \(
    \sat{\PCTLexp{\Pi}} = \textsc{SatE}\left(\Pi\right)
  \).
\end{theorem}

Again, we must ask whether the resulting satisfaction set is representable.
For the diagonalizable case, we first compute the nullspace of the modal matrix, which results in a convex set.
In lines~\ref{line:set-to-0} and~\ref{line:element-H}, we intersect with sets that are convex themselves.
Thus, the resulting sets are also convex.
The approach of the non-diagonalizable case boils down to solving linear equation systems and computing set membership of a convex polytope, similar to the step-bounded algorithm.
Thus, the satisfaction set remains representable as well.

\begin{figure}[t]%
  \begin{subfigure}[t]{0.5\linewidth}
    \input{res/ex-EH.pgf}%
    \caption{\(\sat{\PCTLexp{\Pi}}\)}%
    \label{fig:ex-satset-EH}%
  \end{subfigure}%
  \begin{subfigure}[t]{0.5\linewidth}
    \input{res/ex-until.pgf}%
    \caption{\(\sat{\PCTLprob{\geq 0.3}{\PCTLbuntil{\Lambda}{\PCTLexp{\Pi}}{5}}}\)}%
    \label{fig:ex-satset-until}%
  \end{subfigure}%
  \caption{A visualization of the satisfaction sets of \(\PCTLexp{\Pi}\) and \(\PCTLprob{\geq 0.3}{\PCTLbuntil{\Lambda}{\PCTLexp{\Pi}}{5}}\).}\label{fig:ex-satset}
\end{figure}

\begin{example}\label{ex:sat-E-formulas}
  The satisfaction sets of the formulas \(\PCTLexp{\Pi}\) and \(\PCTLprob{\geq 0.3}{\PCTLbuntil{\Lambda}{\PCTLexp{\Pi}}{5}}\) as defined in \autoref{ex:E-formulas} are visualized in \autoref{fig:ex-satset}.
  The colour represents the corresponding mode of the Markov chain.%
  \footnote{%
    This example has been implemented and analysed using Wolfram Mathematica with version~14.3.
    While Mathematica is able to compute exact results for~\(\PCTLexp{\Pi}\), the exact analysis of the until operator was aborted after 120 minutes.
    This is likely due to the exponential increase of possible sets per time step that need to be backpropagated (and afterwards intersected).
    We obtained the final results with a numerical analysis instead, with an (initial) \(100\)-digit precision.
  }
\end{example}

\subsection{Adaptions for Model Checking \(\PCTLbvar{k}{\Xi}\) and \(\PCTLvar{\Xi}\)}

\begin{example}\label{ex:sat-V-formulas}
  There is no initial state for which the second moment converges, except trivially the origin.
  Consequently, depending on whether \(\begin{bsmallmatrix}0 & 0 & 0\end{bsmallmatrix}\transpose \in \Xi\), \(\sat{\PCTLvar{\Xi}}\) is either \(\set{(\begin{bsmallmatrix}0 \\ 0\end{bsmallmatrix}, 1), (\begin{bsmallmatrix}0 \\ 0\end{bsmallmatrix}, 2), (\begin{bsmallmatrix}0 \\ 0\end{bsmallmatrix}, 3)}\) or \(\emptyset\).
\end{example}

Effectively, computing the satisfaction sets for the second moment operators \(\PCTLbvar{k}{\Xi}\) and \(\PCTLvar{\Xi}\) is analogous to the aforementioned algorithms, replacing the operator~\(\boperator_{\mathfrak{J}}\) with~\(\toperator_{\mathfrak{J}}\).
This requires some technical adaptations:
Since it propagates the vectorized conditional second moments, the use of \(\toperator_{\mathfrak{J}}\) lifts the whole computation to a dimension of~\(mn^2\) and up to \(mn^2\) eigenvalues are produced.
For the same reason, \(xx\transpose\) has to be vectorized before combining it with~\(\basevect{\theta}\).
Similarly, after computing \(\summn{\mathrlap{mn^2}}{n^2}{}\), the result has to be reshaped to fit the dimension of the polytope~\(\Xi\).
This can be done by first de-vectorizing with~\(\unmvec{n\times n}{}\) and then half-vectorizing the result again with~\(\mvech{}\).
The satisfaction set of \(\PCTLbvar{k}{\Xi}\) appearing in the analogon of~\autoref{thm:correctness-mc-boundedexp} is obtained as follows, with \autoref{alg:sat-long-term-exp} being adapted analogously:
\begin{multline*}
  \sat{\PCTLbvar{k}{\Xi}} = \left\{\xtheta \in \statespace \mid \right.\\
  \left.\mvech{\unmvec{n\times n}{\summn{mn^{2}}{n^{2}}{\frac{1}{k+1}\sum_{i=0}^{k}\toperator_{\mathfrak{J}}^{i}\left(\basevect{\theta} \kron \mvec{xx\transpose}\right)}}} \in \Xi\right\}
\end{multline*}

\subsection{Discussion on Realizability}

This section has laid out the mathematical foundations for the analysis of sta\-bili\-ty-related \gls{pctl} properties.
As is often the case for algorithms of this kind, however, several challenges arise when implementing them in practice.
The most immediate one is that exact arithmetics over irrational numbers is generally not directly realizable.
Even when restricting \glspl{mjls} to rational system dynamics and transition probabilities -- which both may be reasonable assumptions -- some difficulties persist:
In the case of \(\sat{\PCTLexp{\Pi}}\) and \(\sat{\PCTLvar{\Xi}}\), the eigenpairs of \(\boperator\) and~\(\toperator\) have to be computed.
Determining the eigenvalues of a \(k \times k\) matrix is equivalent to finding the roots of a polynomial of degree \(k\), for which no general closed-form solution exists when \(k > 4\) (see Abel-Ruffini theorem~\cite{dummit2003algebra}).
This frontier is exceeded easily, for instance already for \(m = 3\) and \(n = 2\) in first moment analysis and for \(m = 2\) and \(n = 2\) in second moment analysis.
As a consequence, an exact analytical computation becomes impossible for some systems, and the analysis results can only be as accurate as the underlying numerical approximation.

\section{Conclusion}

In this paper, we have addressed \gls{pctl} model checking for \gls{mjls}, enriched with propositions that allow reasoning about stability in a way that can take into consideration specific initial state sets. This enables a novel view on \glspl{mjls} that so far has been impossible to express and decide, based on characterizations of the convergence of the first and second moments, relative to a given initial state.
For the \gls{pctl} fragment excluding the until operator, we have provided a straightforward decision algorithm, while we have highlighted that the general model-checking problem for step-unbounded properties is open and closely related to other longstanding open problems.
For the logic enriched with stability operators (until aside), we have succeeded in developing a corresponding model-checking procedure that harvests linear-algebraic techniques. 
More broadly, our results show that probabilistic model checking provides a natural framework for analysing state-dependent stability properties of \glspl{mjls} beyond classical global stability notions.
Due to the setup of our logic, stability characteristics can be nested inside ordinary \gls{pctl} formulas, but the converse is not supported.
While it is possible to define an extended \gls{pctl} semantics also covering this variant, the model-checking problem for the resulting logic needs to be considered open.
This is an interesting problem for future work.

\begin{credits}
 \subsubsection{\ackname}
  This project has received funding from the European Union's Horizon 2020 research and innovation programme under the Marie Skłodowska-Curie grant agreement \href{https://cordis.europa.eu/project/id/101008233}{No 101008233} -- \href{https://mission-project.eu}{MISSION}, the European Union under the INTERREG North Sea project STORM\_SAFE of the European Regional Development Fund, and DFG grant 389792660 as part of \href{https://cpec.science}{TRR~248 -- CPEC}.
\end{credits}

\bibliographystyle{splncs04}
\bibliography{bibliography.bib}

\clearpage
\appendix
\section{Appendix}\label{sec:appendix-proofs}

\subsection{Example \gls{mjls}}\label{sec:example-mjls}

The running example in \autoref{fig:running-example} has been crafted to illustrate many of the phenomena in a tiny model.
The precise numbers used are as follows:
\begin{align*}
  A_{1} &= \begin{bmatrix}\sfrac{300}{847} & \sfrac{494649}{36784} \\ \sfrac{300}{847} & \sfrac{-474583}{62073}\end{bmatrix} \\
  A_{2} &= \begin{bmatrix}\sfrac{-119633}{6776} & \sfrac{300}{847} \\ \sfrac{102467}{7623} & \sfrac{300}{847}\end{bmatrix} \\
  A_{3} &= \begin{bmatrix}\sfrac{300}{847} & \sfrac{3416121}{1287440} \\ \sfrac{300}{847} & \sfrac{-300}{847}\end{bmatrix} \\
  P &= \begin{bmatrix}
    \sfrac{1}{7} & \sfrac{3}{7} & \sfrac{3}{7} \\
    \sfrac{1}{4} & \sfrac{1}{2} & \sfrac{1}{4} \\
    \sfrac{1}{3} & \sfrac{1}{3} & \sfrac{1}{3}
  \end{bmatrix}
\end{align*}

\subsection{Proofs}

\repeattheorem{thm:pctl-measurability}

\begin{proof}
  Recall that \(\Sigma^{\mathfrak{J}}\) contains all cylinder sets \(\cylset{\hat{\pi}}\) of all finite paths \(\hat{\pi}\) in \(\mathfrak{J}\).
  There are three possibilities for the shape of \(\varphi\): \(\PCTLnext{\Phi}\), \(\PCTLbuntil{\Phi_1}{\Phi_2}{k}\), and \(\PCTLuntil{\Phi_1}{\Phi_2}\).
  If \(\varphi = \PCTLnext{\Phi}\), then:
  \begin{align*}
    &\phantom{{}={}} \set{\pi \in \paths{\xtheta} \where \pi \vDash \PCTLnext{\Phi}}                                                                        \\
    &= \set{\big. (x, \theta)(\statematrix{\theta}x, \theta')\cdots \in \paths{\xtheta} \where (\statematrix{\theta}x, \theta') \vDash \Phi}      \\
    &= \bigcup_{\theta' = 1}^{m} \set{\cylset[\big][\big]{\xtheta (\statematrix{\theta}x, \theta')} \where (\statematrix{\theta}x,\theta') \vDash \Phi} \\
    &\in \Sigma^{\mathfrak{J}}
  \end{align*}
  If \(\varphi = \PCTLbuntil{\Phi_1}{\Phi_2}{k}\), then:
  \begin{align*}
     &\phantom{{}={}} \set{\pi \in \paths{\xtheta} \where \pi \vDash \PCTLbuntil{\Phi_1}{\Phi_2}{k}}                                                                                                                                                                                                                              \\
     &= \bigcup_{j = 0}^{k} \set{\big. (x_{0}, \theta_{0})(x_{1}, \theta_{1})\cdots \in \paths{\xtheta} \where (x_{j}, \theta_{j}) \vDash \Phi_2 \land \forall i < j.\, (x_{i}, \theta_{i}) \vDash \Phi_{1}}                                                                                                     \\
     &= \bigcup_{j = 0}^{k} \, \bigcup_{\hspace{-11pt}\mathrlap{(x_{0}, \theta_{0})\cdots(x_{j}, \theta_{j}) \in \paths[j]{\xtheta}}} \, \set{\cylset[\big][\big]{(x_{0}, \theta_{0}) \cdots (x_{j}, \theta_{j})} \where (x_{j}, \theta_{j}) \vDash \Phi_2 \land \forall i < j.\, (x_{i}, \theta_{i}) \vDash \Phi_{1}} \\
     &\in \Sigma^{\mathfrak{J}}
  \end{align*}
  The case \(\varphi = \PCTLuntil{\Phi_1}{\Phi_2}\) follows directly from the case \(\varphi = \PCTLbuntil{\Phi_1}{\Phi_2}{k}\) by replacing the upper limit~\(k\) of the union by \(\infty\). \qed
\end{proof}

\repeatlemma{lem:gmp-equiv-algo}

\begin{proof}
  Induction on \(\Phi\).
  Most cases are trivial.
  As the only cases of note, we consider \(\PCTLprob{J}{\PCTLnext{\Phi}}\), \(\PCTLprob{J}{\PCTLbuntil{\Phi_{1}}{\Phi_{2}}{k}}\), and \(\PCTLprob{J}{\PCTLuntil{\Phi_{1}}{\Phi_{2}}}\).

  Let \(\Phi = \PCTLprob{J}{\PCTLnext{\Phi}}\).
  The claim \(\xtheta \vDashGMP \PCTLprob{J}{\PCTLnext{\Phi}}\) iff~\({\xtheta \vDash \PCTLprob{J}{\PCTLnext{\Phi}}}\) is equivalent to \(\Q{\xtheta}{\sat{\Phi}} \in J\) iff \(\prob[\big][\big]{\set{\big. \pi \in \paths{\xtheta} \where \pi \vDash \PCTLnext{\Phi}}} \in J\) by definition.
  We prove this equivalence by showing equality of the two probabilities:
  \begin{align*}
    \Q{\xtheta}{\sat{\Phi}}
     & = \sum_{\theta' = 1}^{m} \indicator{\sat{\Phi}}{(\statematrix{\theta}x, \theta')} \transprob{\theta}{\theta'}      & \eqref{eq:gmp-lifting}                                       \\
     & = \sum_{\theta'=1}^{m} \indicator{\sat{\Phi}}{(\statematrix{\theta}x, \theta')} \prob{\cylset{\xtheta(\statematrix{\theta}x, \theta')}}       \\
     & = \prob{\bigsqcup_{\hspace{-11pt}\mathrlap{(\statematrix{\theta} x, \theta') \in \sat{\Phi}}} \cylset{\xtheta(\statematrix{\theta}x, \theta')}}       \\
     & = \prob{\bigsqcup_{\hspace{-11pt}\mathrlap{(\statematrix{\theta} x, \theta') \vDash \Phi}} \cylset{\xtheta(\statematrix{\theta}x, \theta')}}       \\
     & = \prob{\set{\pi \in \paths{\xtheta} \where \pi[1] \vDash \Phi}} \\
     & = \prob{\set{\pi \in \paths{\xtheta} \where \pi \vDash \PCTLnext{\Phi}}}
  \end{align*}

  For \(\Phi = \PCTLprob{J}{\PCTLbuntil{\Phi_{1}}{\Phi_{2}}{k}}\), we show
  \[
    V_{k}\left(\xtheta\right) = \prob{\set{\pi \in \paths{\xtheta} \where \pi \vDash \PCTLbuntil{\Phi_{1}}{\Phi_{2}}{k}}}
  \] by induction on \(k\).
  Let \(\delta_{y}{A} = \indicator{A}{y}\) denote the Dirac measure on a set~\(Y\) for a given \(y\in Y\) and any measurable set~\(A \subseteq Y\).
  For the Dirac measure, the Lebesgue integral \(\int_{Y} f(\xi) \delta_{y}(d\xi)\) is equal to \(f(y)\).
  The base case as well as cases \(\xtheta \in \sat{\Phi_{2}}\) and \(\xtheta \notin \sat{\Phi_{1}} \land \xtheta \notin \sat{\Phi_{2}}\) are trivial.
  Suppose that \(\xtheta \in \sat{\Phi_{1}} \setminus \sat{\Phi_{2}}\), where we note that
  \begingroup
  \allowdisplaybreaks
  \begin{align*}
    &\hphantom{{}={}}V_{k+1}\left(\xtheta\right)\\
     & =\int_{\statespace} V_{k}\left(\xi\right) \cdot \Q{\xtheta}{d\xi}                                                                                                                                   \\
     & = \int_{\statespace} V_{k}\left(\xi\right) \cdot \sum_{\theta'=1}^{m} \indicator{d\xi}{(\statematrix{\theta}x, \theta')} \cdot \transprob{\theta}{\theta'}                                          \\
     & = \sum_{\theta'=1}^{m} \transprob{\theta}{\theta'} \cdot \int_{\statespace} V_{k}\left(\xi\right) \cdot \indicator{d\xi}{(\statematrix{\theta}x, \theta')}                                           \\
     & = \sum_{\theta'=1}^{m} \transprob{\theta}{\theta'} \cdot \int_{\statespace} V_{k}\left(\xi\right) \cdot \delta_{(\statematrix{\theta}x, \theta')}(d\xi) \\ 
     & = \sum_{\theta'=1}^{m} \transprob{\theta}{\theta'} V_{k}\left((\statematrix{\theta}x,\theta')\right)                                                         \\
     & \stackrel{\text{IH}}{=} \sum_{\theta'=1}^{m} \transprob{\theta}{\theta'} \cdot \prob{\set{\pi \in \paths{(\statematrix{\theta}x,\theta')} \where (\statematrix{\theta}x,\theta') \vDash \PCTLbuntil{\Phi_{1}}{\Phi_{2}}{k}}} \\
     & = \prob{\set{\pi \in \paths{\xtheta} \where \pi \vDash \PCTLbuntil{\Phi_{1}}{\Phi_{2}}{k+1}}}
  \end{align*}
  \endgroup
  For \(\Phi = \PCTLprob{J}{\PCTLuntil{\Phi_{1}}{\Phi_{2}}}\) it is easy to see that
  \begin{align*}
    \lim_{k\to\infty} V_{k}\left(\xtheta\right) 
    &= \lim_{k\to\infty} \prob{\set{\pi \in \paths{\xtheta} \where \pi \vDash \PCTLbuntil{\Phi_{1}}{\Phi_{2}}{k}}} \\
    &= \prob{\bigcup_{k=0}^{\infty} \set{\pi \in \paths{\xtheta} \where \pi \vDash \PCTLbuntil{\Phi_{1}}{\Phi_{2}}{k}}} \\
    &= \prob{\set{\pi \in \paths{\xtheta} \where \pi \vDash \PCTLuntil{\Phi_{1}}{\Phi_{2}}}}\tag*{\qed}
  \end{align*}
\end{proof}

\repeattheorem{thm:correctness-mc-boundedpctl}

\begin{proof}
  We show that \(\set{\xtheta \in \statespace \where \xtheta \vDash \Phi}\) (\(= \sat{\Phi}\)) is equal to the sets given above by induction on \(\Phi\).
  The base cases \(\ttrue\) and \(a\), as well as the cases \(\Phi_{1} \land \Phi_{2}\), and \(\neg\Phi'\) are trivial and will be omitted here.

  For \(\Phi = \PCTLprob{J}{\PCTLnext{\Phi'}}\), we have to show that
  \begin{multline*}
    \prob{\set{\pi \in \paths{\xtheta} \where \pi \vDash \PCTLnext{\Phi'}}} \in J \\ \iff \sum_{\theta'=1}^{m}\indicator{\sat{\Phi'}}{(\statematrix{\theta} x, \theta')} \transprob{\theta}{\theta'} \in J
  \end{multline*}
  For this, note that
    \begin{align*}
      &\phantom{{}={}} \prob{\set{\pi \in \paths{\xtheta} \where \pi \vDash \PCTLnext{\Phi'}}}\\
      & = \prob{\set{\pi \in \paths{\xtheta} \where \pi[1] \vDash \Phi'}} \\
      & = \prob{\bigsqcup_{\hspace{-11pt}\mathrlap{(\statematrix{\theta} x, \theta') \vDash \Phi'}} \cylset{\xtheta(\statematrix{\theta}x, \theta')}}       \\
      & = \prob{\bigsqcup_{\hspace{-11pt}\mathrlap{(\statematrix{\theta} x, \theta') \in \sat{\Phi'}}} \cylset{\xtheta(\statematrix{\theta}x, \theta')}}       \\
      & = \sum_{\theta'=1}^{m} \indicator{\sat{\Phi'}}{(\statematrix{\theta}x, \theta')} \prob{\cylset{\xtheta(\statematrix{\theta}x, \theta')}}       \\
      & = \sum_{\theta' = 1}^{m} \indicator{\sat{\Phi'}}{(\statematrix{\theta}x, \theta')} \transprob{\theta}{\theta'}  
  \end{align*}

  For \(\Phi = \PCTLprob{J}{\PCTLbuntil{\Phi_{1}}{\Phi_{2}}{k}}\), we have to show that
  \begin{multline*}
    \forall k \in \NN.\, \forall \xtheta.\, \prob{\set{\pi \in \paths{\xtheta} \where \pi \vDash \PCTLbuntil{\Phi_{1}}{\Phi_{2}}{k}}} \in J \\
    \iff U_{k}\left(\xtheta\right) \in J
  \end{multline*}
  by induction on \(k \in \NN\).
  The base case is trivial.
  For \(k > 0\), we consider the three cases
  \begin{enumerate*}[(i)]
    \item\label{it:sat2} \(\xtheta \vDash \Phi_2\),
    \item\label{it:notsat1notsat2} \(\xtheta \nvDash \Phi_1 \land \xtheta \nvDash \Phi_2\), and
    \item\label{it:sat1notsat2} \(\xtheta \vDash \Phi_1 \land \xtheta \nvDash \Phi_2\).
  \end{enumerate*}

  For cases \ref{it:sat2} and \ref{it:notsat1notsat2}, \(\prob{\set{\pi \in \paths{\xtheta} \where \pi \vDash \PCTLbuntil{\Phi_{1}}{\Phi_{2}}{k}}}\) reduces to \(\prob{\paths{\xtheta}} = 1\) and \(\prob{\emptyset} = 0\), respectively.
  These results coincide with \(U_{k}\left(\xtheta\right)\).

  For \ref{it:sat1notsat2}, note that
  \begingroup
  \allowdisplaybreaks
  \begin{align*}
    &\phantom{{}={}} \prob{\set{\pi \in \paths{\xtheta} \where \pi \vDash \PCTLbuntil{\Phi_{1}}{\Phi_{2}}{k}}}\\
    &= \prob{\set{\pi \in \paths{\xtheta} \where \exists j \leq k.\, \pi[j] \vDash \Phi_{2} \land \forall 0 \leq i < j.\, \pi[i] \vDash \Phi_{1}}}\\
    &= \prob{\bigcup_{j=0}^{k}\set{\pi \in \paths{\xtheta} \where \pi[j] \vDash \Phi_{2} \land \forall 0 \leq i < j.\, \pi[i] \vDash \Phi_{1}}}\\
    &\stackrel{\mathclap{\ref{it:sat1notsat2}}}{=} \prob{\bigcup_{j=1}^{k}\set{\pi \in \paths{\xtheta} \where \pi[j] \vDash \Phi_{2} \land \forall 0 \leq i < j.\, \pi[i] \vDash \Phi_{1}}}\\
    &= \sum_{\theta'=1}^{m} \transprob{\theta}{\theta'} \cdot \mi{Pr}\left(\bigcup_{j=0}^{k-1}\right.
    \begin{multlined}[t]
      \{\pi\in\paths{(A_{\theta}x,\theta')} \;\vert\; {}\\ \left.\vphantom{\bigcup_{j=0}^{k-1}}\pi[j] \vDash \Phi_{2} \land \forall 0 \leq i < j.\, \pi[i] \vDash \Phi_{1}\}\right)
    \end{multlined} \\
    &= \sum_{\theta'=1}^{m} \transprob{\theta}{\theta'} \cdot \mi{Pr}(
      \begin{multlined}[t]
        \{\pi\in\mi{Paths}((A_{\theta}x,\theta')) \;\vert\; \\ \exists j \leq k-1.\,\pi[j] \vDash \Phi_{2} \land \forall 0 \leq i < j.\, \pi[i] \vDash \Phi_{1}\})
      \end{multlined} \\
    &= \sum_{\theta' = 1}^{m} \transprob{\theta}{\theta'} \cdot \prob{\set{\pi \in \paths{(A_{\theta}x,\theta')} \where \pi \vDash \PCTLbuntil{\Phi_{1}}{\Phi_{2}}{k-1}}} \\
    &\stackrel{\text{IH}}{=} \sum_{\theta' = 1}^{m} \transprob{\theta}{\theta'} \cdot U_{k-1}\left((\statematrix{\theta}x,\theta')\right) \\
    &= U_{0}\left(\xtheta\right)
    + \indicator{\sat{\Phi_{1}} \setminus \sat{\Phi_{2}}}{\xtheta} \sum_{\theta' = 1}^{m} \transprob{\theta}{\theta'} \cdot U_{k-1}\left((\statematrix{\theta}x,\theta')\right) \\
    &= U_{k}\left(\xtheta\right) \tag*{\qed}
  \end{align*}
  \endgroup
\end{proof}

\end{document}